\documentclass{article}

\usepackage[preprint]{neurips_2024}
\setcitestyle{aysep={,}} 

\usepackage[utf8]{inputenc} 
\usepackage[T1]{fontenc}    
\usepackage{hyperref}       
\usepackage{url}            
\usepackage{booktabs}       
\usepackage{amsfonts}       
\usepackage{nicefrac}       
\usepackage{microtype}      
\usepackage{xcolor}         
\usepackage{comment} 
\usepackage{amsmath}
\usepackage{amsthm}
\usepackage{algorithmicx,algorithm,algpseudocode}
\usepackage{graphicx,accents}
\newcommand{\ubar}[1]{\underaccent{\bar}{#1}}

\newcommand\prs[1]{\left(#1\right)}
\newcommand\sbk[1]{\left[#1\right]}

\newcommand{\E}{\mathbb{E}}

\newcommand{\gtz}{\rightarrow0}

\newtheorem{theorem}{Theorem}
\newtheorem{assumption}{Assumption}
\newtheorem{corollary}{Corollary}
\newtheorem{proposition}{Proposition}
\theoremstyle{definition}

\theoremstyle{remark}

\title{Double Distributionally Robust Bid Shading\\for First Price Auctions}

\author{%
Yanlin Qu$^1$ \quad Ravi Kant$^2$ \quad Yan Chen$^3$ \quad Brendan Kitts$^2$ \\
\textbf{San Gultekin}$^2$ \quad \textbf{Aaron Flores}$^2$ \quad \textbf{Jose Blanchet}$^1$ \\
$^1$Stanford University \quad $^2$Yahoo! \quad $^3$Duke University\\
\texttt{\{quyanlin,jose.blanchet\}@stanford.edu}\\
\texttt{\{rkant,brendan.kitts,sgultekin,aaron.flores\}@yahooinc.com}\\
\texttt{yan.chen@duke.edu}
}

\begin{document}

\maketitle

\begin{abstract}
  Bid shading has become a standard practice in the digital advertising industry, in which most auctions for advertising (ad) opportunities are now of first price type. Given an ad opportunity, performing bid shading requires estimating not only the value of the opportunity but also the distribution of the highest bid from competitors (i.e. the competitive landscape). Since these two estimates tend to be very noisy in practice, first-price auction participants need a bid shading policy that is robust against relatively significant estimation errors. In this work, we provide a max-min formulation in which we maximize the surplus against an adversary that chooses a distribution both for the value and the competitive landscape, each from a Kullback-Leibler-based ambiguity set. As we demonstrate, the two ambiguity sets are essential to adjusting the shape of the bid-shading policy in a principled way so as to effectively cope with uncertainty. Our distributionally robust bid shading policy is efficient to compute and systematically outperforms its non-robust counterpart on real datasets provided by Yahoo DSP.
\end{abstract}

\section{Introduction}
\label{introduction}
Online advertising enables businesses to promote their products and services through placing ads on social media, search engines, apps, and websites. Publishers sell their advertising opportunities through online auctions organized by ad exchanges where advertisers bid in real-time. Demand Side Platforms (DSPs) bid on behalf of advertisers to optimally spend their advertising dollars. The work in this paper was done in collaboration with one of the largest DSPs in the industry, with more than a billion dollars in yearly revenue, providing the data resources at scale for evaluation of the techniques shown in this paper. The Global DSP market size is estimated to be above \$20B \cite{ForBusIn}, with first price auctions constituting almost all of the revenue \cite{ou2023survey}.

Second Price Auctions (SPAs) used to play a dominant role in online advertising. Google, in particular, leveraged SPAs for its Adwords and Adsense platforms, starting in the 1990s. However, the period between 2018 and 2019 marked a significant transition in the digital advertising sector, with a majority of display ad auctions moving away from SPAs to First Price Auctions (FPAs) by 2020. This transition was motivated by various factors, including a growing demand for greater transparency and accountability, concerns over yields, and the inclusion of header bidding technologies that are difficult to implement in the SPA context.

As we shall review, optimal bidding in the SPA setting is easy, whereas optimal bidding in FPAs is difficult due to hard-to-reduce estimation errors. So, the goal of this paper is to provide a disciplined approach based on distributionally robust optimization for optimal surplus bidding in FPAs.

In SPAs, the winner pays the second-highest bid price. SPAs have what is known as the Vickrey property \cite{vickrey1961counterspeculation}; i.e. bidding truthfully is a dominant strategy that leads to the best possible outcome regardless of how competitors bid. Therefore, given an ad opportunity, advertisers should calculate its value and then bid such value.

In FPAs, the winner pays the amount bid. Therefore, it is not optimal to bid truthfully (i.e. the exact value of the opportunity) since in that case there is zero profit. Consequently, in FPAs advertisers have an incentive to bid a price lower than the estimated value of the opportunity. This lowering procedure is called bid shading, which increases the profit margin to the winner but decreases the win rate. 

To address this trade-off, DSPs (on behalf of the advertisers they represent) need to perform two challenging estimation tasks. First, they need to estimate the win rate as a function of bid price, which is the cumulative distribution function (CDF) of the highest bid from competitors (minimum winning price). And, second, they need to estimate the value of an advertisement opportunity.


DSPs receive billions of bid requests per day and send millions of bids per second.
For the two challenging estimation tasks mentioned above, DSPs have abundant data to train complex machine learning models (left half of Figure \ref{flow}), but the output of these models needs to be summarized succinctly in the bid shading policy due to online latency constraints. As a result, the two estimates used in traditional bid shading approaches tend to be very noisy. Therefore, in order to mitigate the impact of these noisy estimates and improve the out-of-sample surplus of bid-shading policies, we propose a novel robust algorithm  (right half of Figure \ref{flow}). While other distributionally robust optimization methods have been introduced for first price auctions \cite{kasberger2023robust}, our method uses two so-called distributionally robust uncertainty sets, one to capture noise estimation in the value of the opportunity and the other one to account for noise estimation in the win rate (corresponding to the two challenging tasks mentioned earlier). 

A conventional bid shading policy maximizes the expected surplus under a baseline model built out of the estimated value and the win rate. In contrast, the distributionally robust bid shading (DRBS) framework introduces a fictitious adversary that wishes to minimize the surplus by choosing a probability model (within a specified neighborhood) for the value of the opportunity and the win rate. The set of probability models that we utilize is based on the Kullback-Leibler (KL) divergence. In our DRBS formulation, we use two KL-based neighborhoods, one corresponding to the value and the second one corresponding to the win rate. The resulting bid shading policy is obtained from a max-min game in which the objective of this zero-sum game is the expected surplus.

\begin{figure}[ht]
\begin{center}
\centerline{\includegraphics[width=0.7\columnwidth]{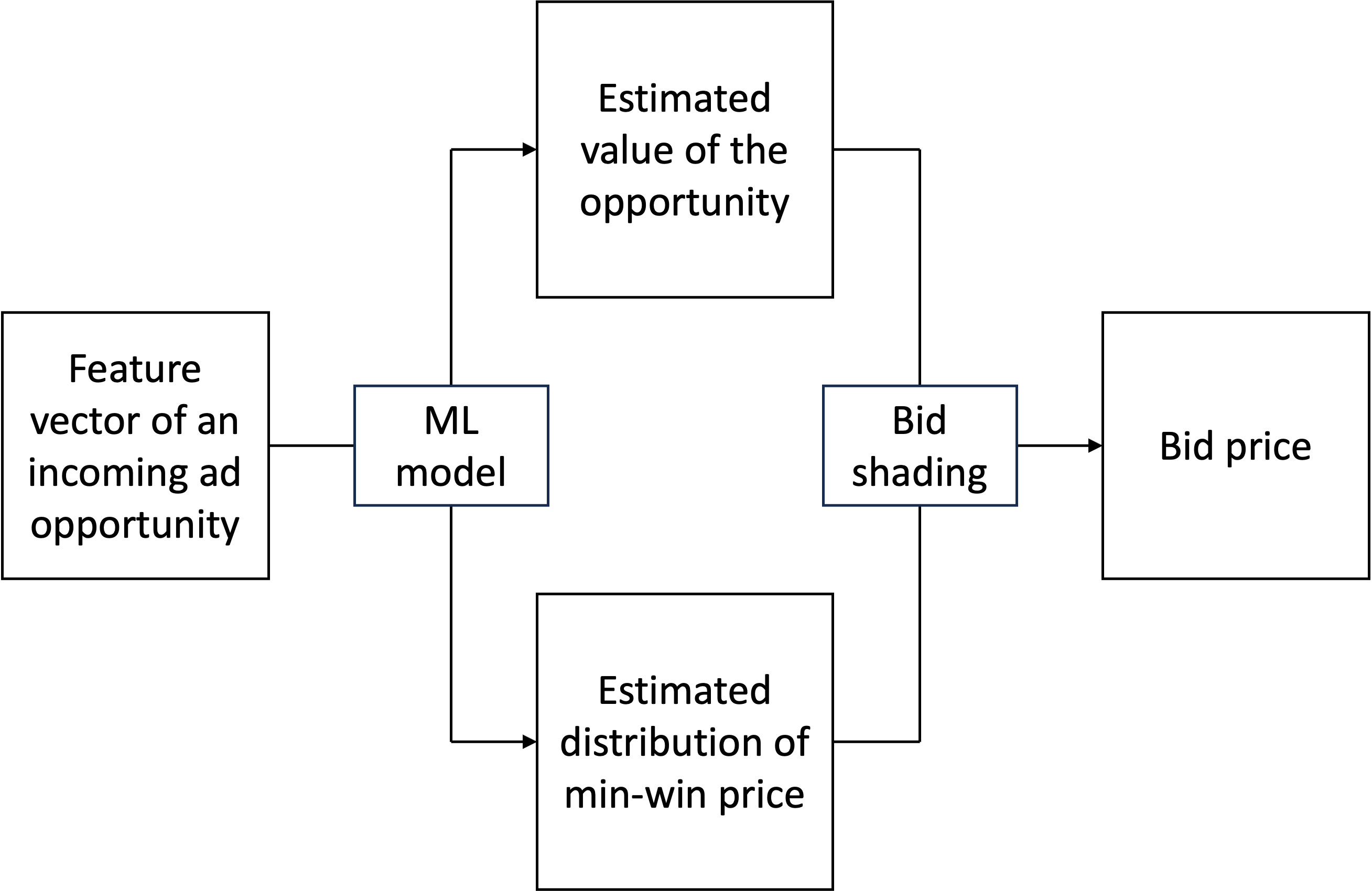}}
\caption{Real-time bidding system for first-price auction}
\label{flow}
\end{center}
\end{figure}
\subsection{Main contributions}

In addition to formulating this DRBS problem with two KL uncertainty sets (the majority of distributionally robust optimization formulations involve only one uncertainty set), our second technical contribution is providing easy-to-compute expressions for the resulting DRBS policy, which are given in our main result in Theorem \ref{main}. The simplicity of the approach is important due to the latency issues mentioned earlier when deploying bid-shading algorithms at scale.

In particular, in Theorem \ref{main} we show that the bid shading policy can be obtained as the root of an explicit monotonic function within a known finite interval. So, it can be efficiently computed (via bisection) under strict latency constraints. 

Further, the resulting bid-shading formula reveals that the two sources of uncertainty affect the optimal bid price in opposite directions. 
\begin{quote}
    \it Everything else equal, the more uncertain we are about the competitive landscape, the higher we bid (provided this is profitable). Similarly, the more uncertain we are about the value, the lower we bid.
\end{quote}

Understanding this trade-off (which our resulting formula does in a principled way) is essential for improving the DSPs surplus. We show rigorously in Section \ref{DRBS}; see Corollary \ref{cor} and Proposition \ref{prop} that our bid-shading strategy allows to control the previous trade-off in a way that is directly interpretable in terms of the degree of distributional uncertainty in the competitive landscape and the value, respectively. 

The improved performance derived from these principles is confirmed by our experiments (Section \ref{experiment results}). The two effects are calibrated via the sizes of the two KL neighborhoods, as we shall explain in Section \ref{experiment design}.

We implement and test the DRBS policy on real datasets provided by Yahoo DSP. We observe that our DRBS systematically outperforms its non-robust counterpart under the effective cost per action (eCPX) metric.



\subsection{Related works}
A comprehensive survey on real-time bidding can be found in \cite{ou2023survey}. Our brief literature review will focus on first-price auctions and distributionally robust optimization (DRO).

For advertisers to better survive {\it the brave new world of first-price auctions} \cite{gligorijevic2020bid}, the literature on bid shading is growing rapidly \cite{despotakis2021first, feng2021reserve,han2020learning,zhang2021meow,grigas2017profit}. In \cite{zhou2021efficient}, neural networks are used to estimate the competitive landscape, and the optimal bid price is obtained by maximizing the expected surplus, which serves as the (non-robust) baseline in our paper. 
The difficulty of estimating the competitive landscape depends on whether the minimum winning price (MWP) data is provided by Supply Side Platforms (SSPs). When the MWP is provided, in \cite{gligorijevic2020bid}, field-weighted factorization machines (FwFM) \cite{rendle2010factorization,guo2017deepfm} are used to predict the optimal shading factor (the bid price divided by the MWP). FwFMs have also been successfully applied to handle sparse data \cite{juan2016field, pan2019predicting,pan2018field}. When the MWP is censored by SSPs (the only available information is whether the MWP is larger than the bid price), the WinRate model in \cite{pan2020bid} still manages to approximate the distribution of MWP via log-logistic model. For multi-layer auctions, the above two classes of approaches are combined in \cite{si2022optimal} to better estimate the competitive landscape.

Regardless of whether SSPs provide the MWP, the optimal bid price should maximize the expected surplus according to the estimated value and the estimated competitive landscape. Given the noisy nature of these estimates, the maximization should greatly benefit from DRO; see \cite{rahimian2022frameworks} for a recent review of DRO. The application of DRO to first-price auctions starts from \cite{kasberger2023robust} where the ambiguity set is large and contains all distributions (of MWP) that have enough mass beyond a small level. The major innovation of our work is introducing two (small) KL ambiguity sets that turn out to be essential in the context of real-time bidding. Recently, holistic DRO is proposed in \cite{bennouna2022holistic}, where the ambiguity set is also described by multiple constraints, to fight against all sources of overfitting.





\section{Distributionally Robust Bid Shading}
\label{DRBS}
\subsection{Problem formulation}
To begin, we recall the (non-robust) baseline bid shading problem. 
Given an ad opportunity, let $X$ be the highest bid from competitors (minimum winning price) and $b$ be our bid price for this opportunity. When the latter is larger, we win the auction and pay for the opportunity. Then our ad is shown to the viewer.
If the ad is click based, it generates value if and only if it is clicked by its viewer.
Let $v$ be the (expected) value of the ad, which is the expectation of the realized value $V$, a scaled Bernoulli random variable ($a\cdot\text{Ber}(p)$).
Let $\bar{P}$ and $\bar{Q}$ be the (estimated) distributions of $V$ and $X$. Given the two distributions, the goal of bid shading is 
\begin{equation}
    \max_{b\geq0}\E_{\bar{P},\bar{Q}} (V-b)I(X\leq b),\tag{baseline}
\end{equation}
i.e. choosing the bid price below $v$ that maximizes the expected surplus.
It is natural to assume that $V$ and $X$ are independent because whether the viewer clicks our ad given $\bar{P}$ is irrelevant to whatever happens behind the screen.
Since the two estimates $\bar{P}$ and $\bar{Q}$ tend to be noisy, we introduce two independent KL balls centered at them as ambiguity sets
\begin{align*}
    \mathcal{P}(\delta_V)=\{P:D(P||\bar{P})\leq\delta_V\},\;\;\mathcal{Q}(\delta_X)=\{Q:D(Q||\bar{Q})\leq\delta_X\},
\end{align*}
where $\delta_V,\delta_X>0$ are radii. Recall that the KL divergence $D$ is given by
$$D(p_1||p_2)=\E_{p_1}\log\prs{p_1(Z)/p_2(Z)}$$
where $p_1$ and $p_2$ can be two PMFs (for $P$ and $\bar{P}$) or two PDFs (for $Q$ and $\bar{Q}$) with $p_1\ll p_2$ ($p_2(x)=0$ implies $p_1(x)=0$).
To maximize the worst-case expected surplus, the distributionally robust bid shading problem is
\begin{equation}
    \max_{b\geq0}\min_{\substack{P\in\mathcal{P}(\delta_V)\\Q\in\mathcal{Q}(\delta_X)}}\E_{P,Q} (V-b)I(X\leq b)\tag{DRBS}
\end{equation}
where the adversary tries to minimize the expected surplus by disturbing the two distributions within the two KL balls.
\subsection{Computable solution}
Surprisingly, the above double DRO problem is almost analytically solvable. To present our solution, some notation and assumptions are needed. Let $\bar{p}=P_{\bar{P}}(V>0)$ be the estimated click probability under $\bar{P}$. Define $\bar{v}=(v/\bar{p})r^{-1}(\delta_V)$ where $r^{-1}$ is the inverse of
$$r(p)=p\log(p/\bar{p})+(1-p)\log((1-p)/(1-\bar{p})),\;\;p\in[0,\bar{p}],$$
which is strictly decreasing.
In fact, this $\bar{v}$ is the worst-case value on $\mathcal{P}(\delta_V)$, which will serve as an upper bound of the DRBS solution in Theorem \ref{main}. After replacing $v$ with $\bar{v}$ in the baseline problem, its solution satisfies (the first order optimality condition)
$$F(b)=(\bar{v}-b)f(b),\;\;b\in[0,\bar{v}]$$
where $F$ and $f$ be the CDF and PDF of $X$ under $\bar{Q}$, respectively. We use $\ubar{v}$ to denote this baseline solution, which will serve as a lower bound of the DRBS solution in Theorem \ref{main}. The existence and uniqueness of $\ubar{v}$ has been discussed in \cite{zhou2021efficient}.
\begin{assumption}
    $\delta_V<r(0)$ and $\delta_X<-\log(1-F(\bar{v})).$
\end{assumption}
This assumption excludes the trivial case of too much uncertainty where the DRBS objective is always zero regardless of $b$. 
\begin{assumption}
    $F$ is log-concave.
\end{assumption}
This is the key assumption that makes the DRBS solution easily computable. Log-concave distributions have log-concave PDFs and hence log-concave CDFs \cite{bagnoli2006log}, i.e. both $\log f$ and $\log F$ are concave. For example, normal and log-normal distributions are log-concave.
\begin{assumption}
    $F(\bar{v})<1/2$ and $F(0)=0.$
\end{assumption}
This is a technical assumption that makes the proof simpler. Given that the overall win rate of the DSP in our experiments is around 5\% (this rate is typical), the assumption holds in practice. 
Now we present the solution of the DRBS problem, assuming that $\bar{v}$ and $\ubar{v}$ have already been computed.
\begin{theorem}
\label{main}
Under Assumption 1-3, the DRBS solution $b^*$ is the unique solution of $g(b)=\delta_X$ in $[\ubar{v},\bar{v}]$
where
\begin{gather*}
g(b)=\log\eta(b)-\log J(b)-\frac{F(b)\log\eta(b)}{J(b)},\\
    J(b)=F(b)+\eta(b)-F(b)\eta(b),\;\;\eta(b)=h^{-1}(L(b))\\
    L(b)=\frac{F(b)}{(\bar{v}-b)f(b)},\;\;h(x)=\frac{x-1}{\log x},\;\;x\geq0.
\end{gather*}
In particular, $g$ is strictly increasing in $[\ubar{v},\bar{v}]$, so $b^*$ can be computed via bisection.
\end{theorem}
We defer the proof of this result to Section \ref{proof}.
The monotonicity of functions $g$ and $r$ brings not only practical convenience but also theoretical insight. Intuitively, a robust bid shading policy should always bid lower than the baseline policy to stay safe when things go wrong. The following result refutes this intuition and gives us a general rule to bid under different sources of uncertainty.
\begin{corollary}
\label{cor}
    Under Assumptions 1-3, the DRBS solution $b^*$ is increasing in $\delta_X$ but decreasing in $\delta_V.$
\end{corollary}
Here we present a possible explanation of this result. When we are uncertain about the competitive landscape (only $\delta_X>0$), the competition is more fierce than expected in the worst case, so we should bid higher than the baseline ($\delta_V,\delta_X=0$) to maintain our win rate. When we are uncertain about the value (only $\delta_V>0$), the opportunity is less valuable than expected in the worst case, so we should bid lower than the baseline to maintain a positive profit margin. When the two sources of uncertainty are present at the same time ($\delta_V,\delta_X>0$), the adjustment made by the DRBS policy can go in either direction, making it a decision maker, deciding whether the baseline policy bids too high or too low. 

Now we argue that the DRBS policy is indeed a reasonable decision maker. The value of an ad opportunity, as the expectation of a scaled Bernoulli random variable, is the product of the click probability and click reward. In general, the click reward of a certain ad is a constant set by its advertiser, so whether an opportunity is valuable to this ad is determined by the click probability. When the click probability is too high or too low, the following result shows that the DRBS policy knows the correct direction in which the baseline policy should be adjusted. 
\begin{proposition}
\label{prop}
    Given an ad opportunity with everything fixed except its click probability $\bar{p}$, the DRBS policy bids higher (or lower) than the baseline policy when $\bar{p}$ is large (or small) enough.
\end{proposition}
We defer the proof of this result to Section \ref{proof}. In the proof, we see that $\delta_X$ (or $\delta_V$) takes the lead to bid higher (or lower) when the click probability is too high (or low), which supports our claim: having two ambiguity sets is essential.
This claim will be further supported when we discuss the experimental design in Section \ref{experiment design}.
The potential decision making wisdom of the DRBS policy indicated in Proposition \ref{prop} will be fully illustrated when we present experimental results in Section \ref{experiment results}.
\subsection{Efficient implementation}
In this section, we discuss how to compute the DRBS solution $b^*$ efficiently. Since it is the root of monotonic function $g-\delta_X$ in $[\ubar{v},\bar{v}]$, it can be efficiently found via bisection, if $g,\ubar{v},\bar{v}$ do not bring much trouble.

For $g$, we briefly talk about the computation of $h^{-1}$, which can be expressed in terms of the Lambert W function that solves $W(z)e^{W(z)}=z$. To be specific,
$$h^{-1}(y)=-y\cdot W_{-1}(-e^{-1/y}/y),\;\;y>1$$
where $W_{-1}$ is the secondary real-valued branch of $W$ \cite{corless1996lambert}. There are efficient implementations of $W$ (e.g. scipy.special.lambertw). In addition, utilizing the bound in \cite{chatzigeorgiou2013bounds}, we have
$$j_{2/3}(y)\leq h^{-1}(y)\leq j_1(y),\;\;y>1$$
where $$j_c(y)=y\prs{1+\sqrt{2}\cdot l(y)+c\cdot l^2(y))},\;\;l(y)=\sqrt{1/y+\log y-1},$$so we may obtain a decent approximation of $h^{-1}$ by properly choosing $c\in[2/3,1]$.

For $\bar{v}$, we use bisection to solve $r(p)=\delta_V$ in $[0,\bar{p}]$. For $\ubar{v}$, we do not need to compute this value.
From the proof of Theorem \ref{main}, we know that $b\geq\ubar{v}\Leftrightarrow L(b)\geq1\Rightarrow g(b)\geq0  $, so we can set $g(b)=0$ whenever $L(b)<1$, which allows us to directly perform bisection in $[0,\bar{v}]$ to find $b^*$.

\section{Experimental Design}
\label{experiment design}
\subsection{Dataset}
We compare the DRBS and baseline policy on a Yahoo DSP private bidding dataset on Google Ad Exchange, which contains the information of two millions of bid requests, with the following fields.
\begin{table}[h]
\caption{Data description}
\label{data description}
\begin{center}
\begin{small}
\begin{tabular}{cp{5.5cm}}
\toprule
Field & Description\\
\midrule
\midrule
line\_id  & ID of the line corresponding to the ad that the DSP wants to win the opportunity for\\
\midrule
ceiling, floor& {Range of allowed bid prices}\\
mu, sigma & {Two estimated lognormal parameters of the distribution of the minimum winning price} \\
click\_prob & {Estimated probability of the ad being clicked}\\
click\_reward & {Reward if the ad is actually clicked}\\
\midrule
value & {Product of the above two}\\
min\_win\_price & {Actual minimum winning price}\\
\bottomrule
\end{tabular}
\end{small}
\end{center}
\end{table}

Different lines, identified by line ID, correspond to different campaigns to display different ads to viewers, so we should hold the competition between the two policies line by line. Given that there are more than one thousand lines in the dataset, those competition results will be summarized into a weighted average.




The distribution of minimum winning price is assumed to be lognormal, and the two parameters are estimated using the method in \cite{zhou2021efficient}. The value, as the expectation of a scaled Bernoulli random variable, is the product of the click probability and click reward.

\subsection{Metric}
The last two rows in Table \ref{data description}, the value and the actual minimum winning price, allow us to evaluate the performance of the two policies. The performance metric we use is
$$R=\frac{\sum_i v_iI(X_i\leq b_i)}{\sum_i b_iI(X_i\leq b_i)}$$
where $v$ is the value, $b$ is the bid price, $X$ is the minimum winning price, and the summation is over all bid requests in one line. The numerator of $R$ is the total value collected from auctions won, while the denominator of $R$ is the total spend to win those auctions, so $R$ can be interpreted as the effective value per dollar spent, which is actually the reciprocal of the effective cost per action (eCPX). We use the reciprocal because it is positively related to the performance. The larger the $R$, the better the performance. For line $k$, let $R^D_k$ and $R^B_k$ be the $R$-values corresponding to the DRBS and baseline policies, respectively. Then $$\Delta R_k=\prs{R^D_k/R^B_k-1}\times100\%$$
is the percentage improvement of the new policy over the old one. We use the spend-weighted average
$$\Delta R=\frac{\sum_{k}s_k\Delta R_k}{\sum_{k}s_k}$$
to represent the overall percentage improvement, where the summation is over all lines, and $s_k$ is the total spend of the DRBS/baseline policy on line $k$. The two policies have the same total spend because of spend equating.

Spend equating is crucial to make the offline comparison meaningful in online sense. Advertisers require the DSP to steadily spend all their budgets in given time periods, so the DSP needs to spend at a certain rate, which is enforced by a control algorithm that modifies $v$'s regularly (e.g. daily). As a result, different policies are regulated by the controller to spend roughly the same on average. 
The only way for one policy to outperform the other is spending the same budget more efficiently (i.e. generating higher surplus per unit of budget). 
To mimic the controller in offline experiments, a standard practice is uniformly modifying all $v$'s for the new policy to make it spend the same as the old policy, which is called spend equating. This modification violates our problem formulation where the two policies are facing the same set of $v$'s. However, this is avoided by balancing $\delta_X$ and $\delta_V$ as they affect the spend oppositely (Corollary \ref{cor}). 

The reader may wonder what is the interpretation of using $\delta_X$ and $\delta_V$, which represent the degree of distributional uncertainty in our formulation, to equalize spending. One way to interpret this mechanism is the following. Advertisers implicitly reveals a risk aversion/appetite profile when they require a certain spend rate (lower spend rate corresponds to higher risk aversion). The controller's goal is to implicitly find the risk aversion coefficient. Thus, the risk/uncertainty parameters that our bid-shading model incorporates, play the same role as the controller in our offline experiments. Thus, there is one degree of freedom (out of the two parameters, $\delta_X$ and $\delta_V$) that we need to select. The other parameter is set to equalize spending. We now explain how the radii are chosen explicitly.

\subsection{Choice of radii}
We randomly split the dataset into a training set and a test set. We compute $\delta_X$ and $\delta_V$ on the training set and compare the two policies on the test set. 
Since the unknown distribution of $X$ seems to contain more uncertainty than the scaled Bernoulli distribution of $V$ (given that the former is highly censored and the later is only a two point distribution), we first choose $\delta_X$ (total spend increases) to maximize the total surplus on the training set $$\sum_i(v_i-b^*_i(\delta_X,0))I(X_i\leq b^*_i(\delta_X,0))$$ then choose $\delta_V$ (total spend decreases) to equate the spend on the training set$$\sum_i b^*_i(\delta_X,\delta_V)=\sum_i b^*_i(0,0)$$
where $b^*_i(0,0)$ corresponds to the baseline policy. 

The resulting DRBS policy, with both $\delta$'s positive, handles the two sources of uncertainty while maintaining the same spend rate as the baseline policy. This is why having two ambiguity sets is essential. If there is only one ambiguity set (one of the two $\delta$'s is zero), the corresponding DRBS policy either always bids higher or always bids lower than the baseline policy (Corollary \ref{cor}), so the approach would consistently either overbid or underbid relative to the non-robust policy. With a balanced pair of positive $\delta$'s, the DRBS policy can be fairly compared with the baseline policy. And also it captures both the implicit risk while balancing overbiding vs underbiding in relation to the value and competitiveness inherent in the ad opportunity.

The same spend on the training set does not guarantee the same spend on the test set.
When the two policies are compared on the test set, spend equating (uniformly modifying all $v$'s) is still needed to make sure that they spend exactly the same, but it will have minimal effect as the spend difference on average has already been removed (assuming that the training set represents the test set well). 

\section{Experimental Results}
\label{experiment results}
\subsection{Overall comparison}
The above strategy can be applied to the whole dataset to find a pair of universal $\delta$'s. It can also be applied line by line to obtain line-specific $\delta's$. We begin with the first approach. On 25\% of the whole dataset, we compute universal $(\delta_X,\delta_V)=(\text{0.065,1.9e-5})$ where $\delta_X$ maximizes the surplus (Figure \ref{surplus_delta}) while $\delta_V$ equates the spend (Figure \ref{spend_difference}). In Figure \ref{surplus_delta}, the DRBS surplus (as a function of $\delta_X$) increases rapidly at the beginning as almost winning bids become winning, while it decreases rapidly at the end as winning bids become overpaid. The curve is noisy in the middle as the two effects compete. In Figure \ref{spend_difference}, the DRBS spend (as a function of $\delta_V$) is smoother as it is irrelevant to noisy $v$'s. Note that when $\delta_V=0$ the DRBS policy spends 40\% more than the baseline policy to maximize the surplus, which will be completely twisted by brute-force spend equating (uniformly modifying all $v$'s). By increasing $\delta_V$, we reduce the overspending in the name of handling the uncertainty in $v$'s, resolving the two problems simultaneously. 
\begin{figure}[ht]
    \centering
    \begin{minipage}{.5\textwidth}
        \centering
        \includegraphics[width=.9\linewidth]{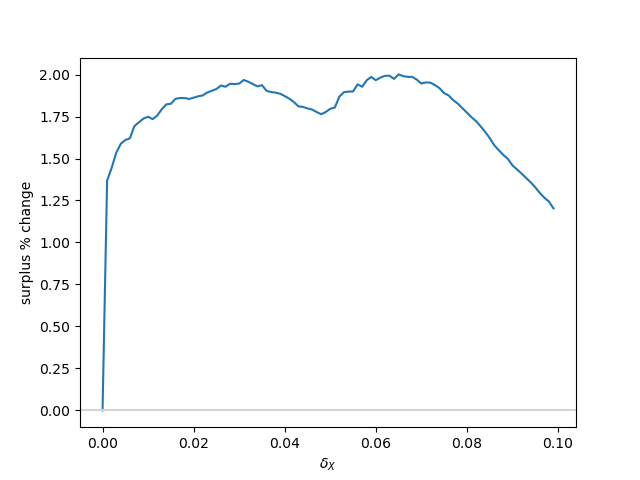}
        \caption{Surplus change} 
        \label{surplus_delta}
    \end{minipage}%
    \begin{minipage}{.5\textwidth}
        \centering
        \includegraphics[width=.9\linewidth]{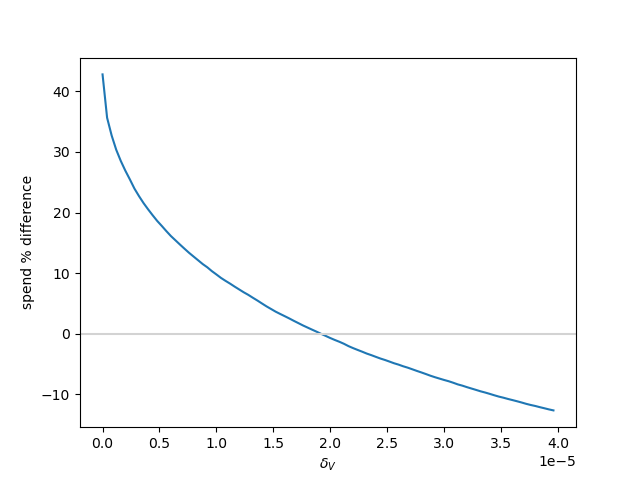}
        \caption{Spend difference}
        \label{spend_difference}
    \end{minipage}
\end{figure}

Equipped with universal $\delta's$, the DRBS policy outperforms the baseline policy by $\Delta R=0.34\%$ on the rest 75\% of the whole dataset. For the five lines with the largest weights, their $\Delta R_k$'s are in the second row of Table \ref{performance} (with superscript $U$).
Universal $\delta$'s favor the two largest lines on the test set, which is natural as they affect $\delta$'s the most on the training set. The performance on the two largest lines confirms that the DRBS policy with properly chosen $\delta$'s can outperforms the baseline policy, while the performance on the other three lines motivates us to consider line-specific $\delta$'s. Equipped with line-specific $\delta$'s, the DRBS policy outperforms the baseline policy by $\Delta R=0.65\%$. For the five lines with the largest weights, their $\Delta R_k$'s are in the third row of Table \ref{performance} (with superscript $L$). 
For lines that benefit from DRBS, line-specific $\delta$'s boost the gain by choosing more suitable $\delta's$.
For lines that do not benefit from DRBS, line-specific $\delta$'s stop the loss by setting $\delta's$ to zero.
\begin{table}[h]
\caption{Performance on the largest 5 lines}
\label{performance}
\begin{center}
\begin{small}
\begin{tabular}{c|ccccc|c}
\toprule
Spend weight & 8.6\% & 8.6\% & 6.6\% & 4.8\% & 4.8\% & Overall\\
\midrule
$\Delta R_k^U$ & 0.8\% & 2.2\% & -1.4\% & -0.3\% & -0.8\% & 0.34\%\\
\midrule
$\Delta R_k^L$ & 1.0\% & 2.3\% & 0.0\% & 0.0\% & 0.0\% & 0.65\%\\
\bottomrule
\end{tabular}
\end{small}
\end{center}
\end{table}
\subsection{V/B analysis}
Recall that $R$ is the ratio between the total value collected from auctions won (sum of $v$'s) and the total spend to win those auctions (sum of $b$'s).
Due to spend equating, it is not allowed to increase (or decrease) the denominator hoping that the numerator increases faster (or decrease slower). The only way to increase $R$ is winning a different set of $(v,b)$-pairs where the sum of $v$'s increases but the sum of $b$'s stays unchanged. Therefore, a good policy should bid higher to win pairs with high $v/b$ while bid lower to lose pairs with low $v/b$ where the former is financed by the latter. 
As Proposition \ref{prop} suggests, the DRBS policy should have the above intuition.
This is confirmed by the following experimental results.

To begin, we focus on the largest line and see whether the DRBS policy exchanges low-$v/b$ wins for high-$v/b$ wins. In Table \ref{difference in wins}, $X$ is the minimum winning price, bid prices $(b^B,b^D)$ are given by the two policies, $v$ is the value, $v/b$ is the ratio between the value and the winning bid price (marked by *). The first five rows are sampled from the opportunities where the DRBS policy wins but the baseline policy loses. The last five rows are sampled from the opportunities where the baseline policy wins but the DRBS policy loses. Clearly, the former has higher $v/b$'s than the latter, which shows that the DRBS policy does exchange low-$v/b$ wins for high-$v/b$ wins. 
In total, the DRBS policy exchange 15121 low-$v/b$ wins with $v/b$-average 1.46 for 1717 high-$v/b$ wins with $v/b$-average 2.36.
Note that $b^D$ is zero when the worst-case value $\bar{v}$ drops below the floor, so the opportunity is no longer worth bidding for.

\begin{table}[ht]
\caption{Difference in wins in the largest line}
\label{difference in wins}
\begin{center}
\begin{small}
\begin{tabular}{p{1cm}p{1cm}p{1cm}p{1cm}p{1cm}p{1cm}}
\toprule
$X$ & $b^B$ & $b^D$ & $v$ & $v/b$\\
\midrule
\midrule
2.47 & 2.29 & 2.49* & 6.64 & 2.67\\
0.59 & 0.58 & 0.60* & 0.89 & 1.48\\
2.09 & 2.00 & 2.20* & 6.83 & 3.10\\
2.96 & 2.93 & 3.11* & 10.25 & 3.30\\
1.78 & 1.77 & 1.86* & 4.22 & 2.27\\
\midrule
0.13 & 0.13* & 0.00 & 0.25 & 1.92\\
0.13 & 0.13* & 0.00 & 0.25 & 1.92\\
0.34 & 0.34* & 0.00 & 0.40 & 1.18\\
0.32 & 0.33* & 0.00 & 0.43 & 1.30\\
0.24 & 0.24* & 0.24 & 0.39 & 1.63\\
\bottomrule
\end{tabular}
\end{small}
\end{center}
\end{table}

To further illustrate the exchange made by the DRBS policy, we view opportunities won as $(v,b)$ pairs and plot them with different colors indicating which policy wins. This is done for the 15 largest lines in Figure \ref{multiplot} where it is clear that the DRBS policy exchanges small-slope points (orange) for large-slope points (blue) in most lines. Proposition \ref{prop} provides a theoretical justification for these smart exchanges.

\begin{figure*}[ht]
\begin{center}
\centerline{\includegraphics[width=1.2\columnwidth]{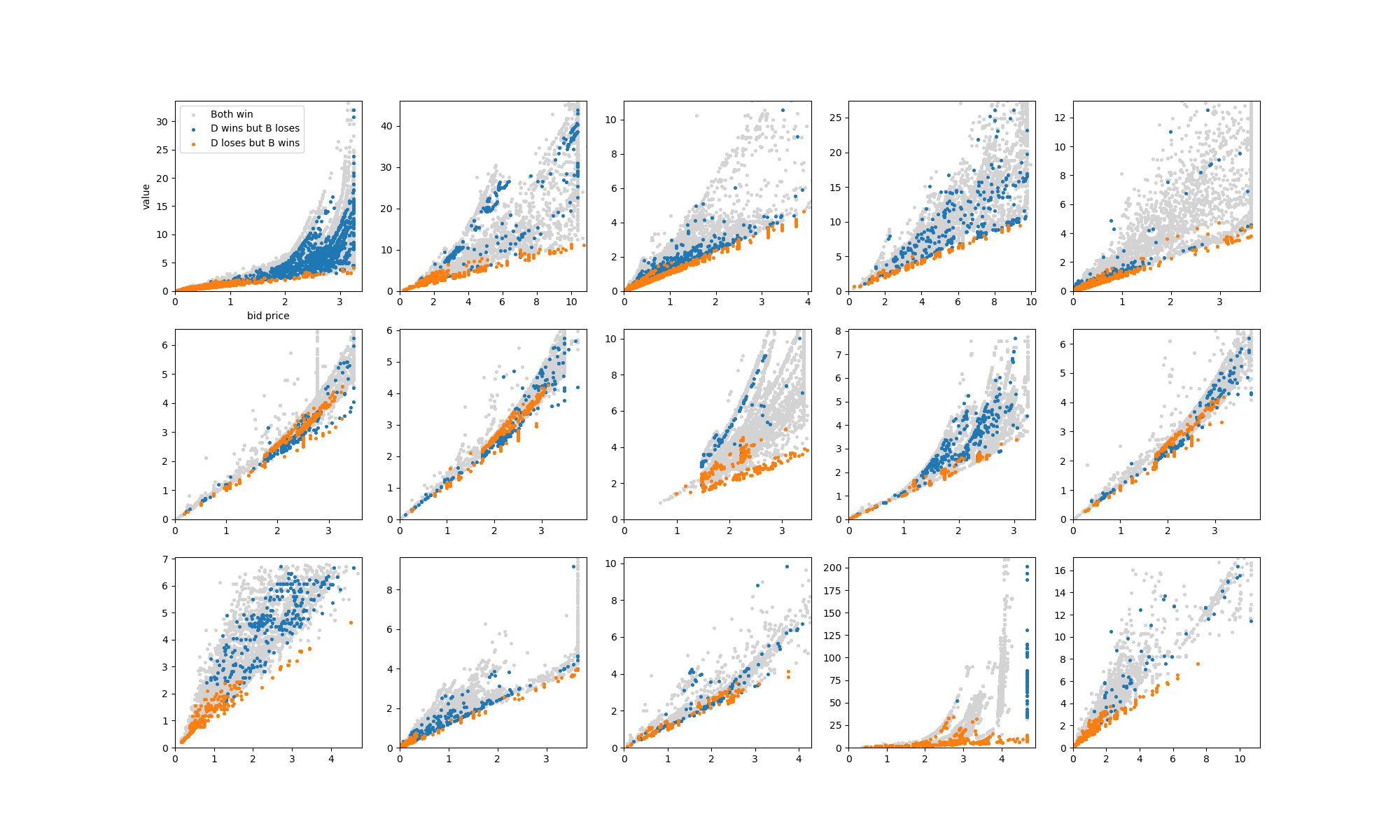}}
\caption{Difference in wins in the largest 15 lines (D: DRBS, B: baseline)}
\label{multiplot}
\end{center}
\end{figure*}

\section{Conclusions}
We presented a distributionally robust bid shading approach built from two KL-based distributional uncertainty regions which accounts for noise estimation in valuation and win rate. We show that the resulting max-min game in which the bidder and adversary wish to optimize the expected surplus can be computed explicitly, resulting in the proposed bid shading policy. 

Further, introducing two uncertainty regions is essential because noise estimation in value and win rate, respectively, affect the optimal bidding in opposite directions (everything else equal, if the competition is more uncertain and it is profitable to bid higher, we should increase the bid; similarly, an uncertain valuation demands more caution and thus a lower bid). This is captured directly by the resulting formula in our max-min DRO formulation. The size of each uncertainty region can be directly interpreted as the impact of the noise in each of these two sources of uncertainty. We further explain how to calibrate the distributional uncertainty size in our experiments, equating spending in order to fairly compare with the non robust setting but also balancing the risk trade-offs mentioned earlier. We believe that this type of multi-layer DRO procedure (i.e. multiple uncertainty regions) has wide applicability in other areas, not only online advertising, and believe that the study of these types of extensions is worth of future research.

\section{Proofs}
\label{proof}
\begin{proof}[Proof of Theorem \ref{main}]
Since $V$ and $X$ are independent,
\begin{align*}
    &\min_{\substack{P\in\mathcal{P}(\delta_V)\\Q\in\mathcal{Q}(\delta_X)}}\E_{P,Q} (V-b)I(X\leq b)=\min_{P\in\mathcal{P}(\delta_V)}\E_P (V-b)\min_{Q\in\mathcal{Q}(\delta_X)}P_Q(X\leq b).
\end{align*}
For the first term, recall that $V$ is a scale Bernoulli random variable, i.e. $V=a\cdot\mathrm{Ber}(p)$ for some $a>0$ and $p\in[0,1].$ In particular, $v=\E_{\bar{P}}V=\bar{a}\bar{p}.$ Then
\begin{align*}
    \bar{v}=&\min_{P\in\mathcal{P}(\delta_V)}\E_PV=\min_{p\in[0,1]:r(p)\leq\delta_V}\bar{a}p=(v/\bar{p})r^{-1}(\delta_V),
\end{align*}
where we use two facts: all $P$'s in $\mathcal{P}(\delta_V)$ have the same support, and $r(p)$ is strictly decreasing in $[0,\bar{p}]$.
By Assumption 1, $\bar{v}>0.$
For the second term, by Theorem 4 in \cite{hu2012kullback}, we have
\begin{align*}
    &\min_{Q\in\mathcal{Q}(\delta_X)}P_Q(X\leq b)\\
    =&\max_{\lambda\leq0}\sbk{\lambda \delta_X+\lambda\log\prs{\E_{\bar{Q}}e^{I(X\leq b)/\lambda}}}\\
    =&\max_{\lambda\leq0}\lambda \sbk{\delta_X+\log\prs{1+F(b)(e^{1/\lambda}-1)}}.
\end{align*}
By Assumption 1 and 
$$\inf_{\lambda\leq0\leq b\leq\bar{v}}F(b)(e^{1/\lambda}-1)=-F(\bar{v}),$$
the above maximum is strictly positive for some $b\in(0,\bar{v})$.
Now the DRBS problem becomes
\begin{align*}
    0<&\max_{b\geq0}\min_{\substack{P\in\mathcal{P}(\delta_V)\\Q\in\mathcal{Q}(\delta_X)}}\E_{P,Q} (V-b)I(X\leq b)\\
    =&\max_{b\geq0}(\bar{v}-b)\min_{Q\in\mathcal{Q}(\delta_X)}P_Q(X\leq b)\\
    =&\max_{\lambda\leq0\leq b}(\bar{v}-b)\lambda\sbk{\delta_X+\log\prs{1+F(b)(e^{1/\lambda}-1)}}.
\end{align*}
From the last line, we know that the maximum can not be achieved when $\lambda=0.$
Since $F(0)=0$ (Assumption 3), the maximum can not be achieved when $b=0$ either. Therefore, the optimizer $(b^*,\lambda^*)$ must satisfy the first order optimality condition
\begin{align*}
    0=&\delta_X+\log\prs{1+F(b)(e^{1/\lambda}-1)}-\frac{F(b)e^{1/\lambda}(1/\lambda)}{1+F(b)(e^{1/\lambda}-1)},\\
    0=&-\sbk{\delta_X+\log\prs{1+F(b)(e^{1/\lambda}-1)}}+(\bar{v}-b)\frac{f(b)(e^{1/\lambda}-1)}{1+F(b)(e^{1/\lambda}-1)}.
\end{align*}
Let $\eta=e^{-1/\lambda}\geq1$. Then
\begin{align*}
    0=&\delta_X+\log\prs{1+F(b)(1/\eta-1)}+\frac{F(b)\log\eta}{\eta+F(b)(1-\eta)},\\
    0=&F(b)\log\eta+(\bar{v}-b)f(b)(1-\eta).
\end{align*}
Let $h(\eta)=(\eta-1)/\log \eta$. Then the second equation becomes
\begin{align*}
    \eta=h^{-1}(L(b))=h^{-1}\prs{\frac{F(b)}{(\bar{v}-b)f(b)}},
\end{align*}
and the first equation becomes
\begin{align*}
    \delta_X=g(b)=&\log\eta(b)-\log\prs{F(b)+\eta(b)-F(b)\eta(b)}-\frac{F(b)\log\eta(b)}{F(b)+\eta(b)-F(b)\eta(b)}.
\end{align*}
Since $F$ is log-concave (Assumption 2), 
$$L(b)=\frac{1}{(\bar{v}-b)(\log F(b))'}$$
is strictly increasing in $[0,\bar{v}]$.
Recall that $\ubar{v}$ is the solution of $L(b)=1$.
Since $h$ is strictly increasing in $[0,\infty)$, 
we have $L(b^*)=h(\eta(b^*))\geq h(1)=1$ that leads to $b^*\geq\ubar{v}.$
The next task is proving that $g$ is strictly increasing in $[\ubar{v},\bar{v}]$.
Since $h^{-1}$ is strictly increasing in $[0,\infty)$,
$\eta(b)=h^{-1}(L(b))$ is strictly increasing in $[0,\bar{v}]$. Therefore, it suffices to show that $g$ is strictly increasing as a function of $\eta\geq1$.
We rewrite $g$ as
\begin{align*}
    g=&\frac{(1-F)\eta\log\eta}{1+(1-F)(\eta-1)}-\log(1+(1-F)(\eta-1))
\end{align*}
From $h(\eta(b))=F(b)/[(\bar{v}-b)f(b)],$
we know that $F$ is a strictly increasing function of $\eta$.
Let $x=\eta-1\geq0$ and $y=1-F\in[1-F(\bar{v}),1-F(\ubar{v})]$. 
Since $F(\bar{v})<1/2$ (Assumption 3), we have $y'<0$ and $1/2<y\leq1.$
The goal is to show that
$$\tilde{g}(x)=\frac{y(1+x)\log(1+x)}{1+xy}-\log(1+xy)$$
is strictly increasing in $x$. We compute its derivative
\begin{align*}
    \tilde{g}'(x)=&\frac{1}{1+x}\frac{y(1+x)}{1+xy}+\log(1+x)\prs{\frac{y(1+x)}{1+xy}}'-\frac{y+xy'}{1+xy}\\
    =&-\frac{xy'}{1+xy}+\log(1+x)\frac{y'(1+x)+y}{1+xy}-\log(1+x)\frac{y(1+x)(y+xy')}{(1+xy)^2}\\
    =&\frac{y'(-x(1+xy)+(1+x)\log(1+x))}{(1+xy)^2}+\frac{y(1-y)\log(1+x)}{(1+xy)^2}\\
    >&\frac{y'(-x(1+x/2)+(1+x)\log(1+x))}{(1+xy)^2}\\
    \geq&0,
\end{align*}
where the last inequality is because the function of $x$ in the numerator is non-positive.
\end{proof}
\begin{proof}[Proof of Proposition \ref{prop}]
    As $\bar{p}\gtz$, $\tilde{p}=r^{-1}(\delta_V)$ reaches $0$ earlier. To be specific, $\tilde{p}=0$ when
    $\delta_V\geq r(0)=-\log(1-\bar{p})$. Therefore, when $\bar{p}$ is small enough, the DRBS policy stops bidding but the baseline policy does not. 
    
    As $\bar{p}\rightarrow 1,$ we must have $\tilde{p}\rightarrow 1$, otherwise
    $$\delta_V=r(\tilde{p})=\tilde{p}\log(\tilde{p}/\bar{p})+(1-\tilde{p})\log((1-\tilde{p})/(1-\bar{p}))$$
    will fail. 
    Recall that $v=a\bar{p}$, $\bar{v}=a\tilde{p}$, and the DRBS objective is
    $$D(b;\bar{v})=(\bar{v}-b)\min_{Q\in\mathcal{Q}(\delta_X)}P_Q(X\leq b).$$
    As $\bar{p}\rightarrow1$ and hence $\tilde{p}\rightarrow1$, we have $D(b;\bar{v})\rightarrow D(b;a)$ uniformly in $[0,a]$. Since each continuous $D(b;\bar{v})$ has a unique maximizer $b^*(\bar{v})$, we must have $b^*(\bar{v})\rightarrow b^*(a)$. Recall that the baseline objective is
    $$B(b;v)=(v-b)P_{\bar{Q}}(X\leq b).$$
    As $\bar{p}\rightarrow1$, we have $B(b;v)\rightarrow B(b;a)$ uniformly in $[0,a]$. Since each continuous $B(b;v)$ has a unique maximizer $b_*(v)$, we must have $b_*(v)\rightarrow b_*(a)$. By Corollary \ref{cor}, $b^*(a)>b_*(a)$, so the DRBS policy bids higher than the baseline policy when $\bar{p}$ is large enough.
    
\end{proof}

\bibliographystyle{plainnat}
\bibliography{sample-base}

\end{document}